 \documentclass[11pt,letterpaper]{article}

\usepackage{times,amsmath,amssymb,amsmath,amsthm,amstext,hyperref}
\usepackage{hyperref}

\usepackage{color}   
\usepackage[margin=1in]{geometry}
 \usepackage{multirow}
\usepackage{microtype} 
\usepackage{epsfig}
\usepackage{bm}
 \usepackage{algorithmic}
 \newtheorem{theorem}{Theorem}[section]

\newtheorem{claim}[theorem]{Claim}

\newtheorem{lemma}[theorem]{Lemma}

\newtheorem{note}[theorem]{Note}

\theoremstyle{definition} 
\newtheorem{definition}[theorem]{Definition}

\newcommand{\eqdef}{\stackrel{def}{=}}
\newcommand{\N}{\mathbb{N}}

\newcommand{\Hn}{H(n)} 
\newcommand{\Oh}{\mathcal{O}}

\newcommand{\ceil}[1]{{\lceil{#1}\rceil}}

\newcommand{\OhT}{\Tilde{\Oh}}

\newcommand{\bool}{\{0,1\}}

\newcommand{\Exp}{\mathbb{E}} 

\newcommand{\lld}{\log\log|\Sigma|}

\title{Property Testing of Joint Distributions using Conditional Samples}

  \author{Rishiraj Bhattacharyya \thanks{NISER
    Bhubaneswar, HBNI, India,
   \texttt{rishiraj.bhattacharyya@gmail.com}} \and Sourav Chakraborty \thanks{Chennai Mathematical Institute
   Chennai, India, and CWI Amsterdam, The Netherlands.
   \texttt{sourav@cmi.ac.in} }}
\date{}
\usepackage{booktabs} % For formal tables

\usepackage[ruled]{algorithm2e} % For algorithms

\begin{document} 
\maketitle

\begin{abstract}
  In this paper, we consider the problem of testing properties of joint distributions under the Conditional Sampling framework. In the standard sampling model, the sample complexity of testing properties of joint distributions is exponential in the dimension, resulting in inefficient algorithms for practical use. While recent results achieve efficient algorithms for product distributions with significantly smaller sample complexity, no efficient algorithm is expected when the marginals are not independent.

We initialize the study of conditional sampling in the multidimensional setting. We propose a subcube conditional sampling model where the tester can condition on an (adaptively) chosen subcube of the domain. Due to its simplicity, this model is potentially implementable in many practical applications, particularly when the distribution is a joint distribution over $\Sigma^n$ for some set $\Sigma$.  

We present algorithms for various fundamental properties of distributions in the subcube-conditioning model and prove that the sample complexity is 
polynomial in the dimension $n$ (and not exponential as in the traditional model).  
We present an algorithm for testing identity to a known distribution  using $\OhT(n^2)$-subcube-conditional samples, an algorithm for testing identity between two unknown distributions using  
$\OhT(n^5)$-subcube-conditional samples and an algorithm for testing identity to a product distribution using $\OhT(n^5)$-subcube-conditional samples.
 
The central concept of our technique involves an elegant chain rule which can be proved using basic techniques of probability theory yet powerful enough to avoid the curse of dimensionality.  
\end{abstract}

\section{Introduction}
\label{sec:introduction}

\textsc{Property Testing of Distributions.} The boom of Big Data
Analytics has rejuvenated the well-studied area of hypothesis testing
over unknown distributions. In Computer Science, the study of this
type of problems was initiated by Batu, Fortnow, Rubinfeld, Smith, and White~\cite{BatuCloseness}
under the framework of ``Property Testing'' \cite{GoldreichGR98,RubinfeldS96}  In this framework, the
``tester''  draws independent samples from the distribution, and
decides whether the distribution satisfies a specific property
$\mathcal{P}$ (null hypothesis) or is far from any distribution that
satisfies $\mathcal{P}$ (alternate hypothesis).

Several properties of probability distributions have been studied in
this framework. Testing whether the distribution is
uniform \cite{Identity,uni1,conf/soda/ChanDVV14}, testing identity 
between two unknown distributions (taking samples from both the distributions) \cite{BatuCloseness,collections}, testing
independence of marginals of product distributions \cite{Identity} ,
estimating entropy \cite{ent1} are a few of the numerous problems that
have been studied in the literature. See \cite{Clement15} for a survey on 
results related to distribution testing. 

Unfortunately, from the modern data analytics point of view, the traditional
framework of sampling yields impractical sample complexity. For example, testing if a distribution
over a set of $n$ elements is uniform requires $\Omega(\sqrt{n})$ samples from the distribution. 
The other problems mentioned above have sample complexity at least this high and in some
cases, almost linear in $n$ \cite{supp,VV2010,valiant}.

\subsection*{Conditional Sampling}

%\textbf{\textsc{Conditional Sampling.}}
To remedy this situation, Chakraborty et~al.~\cite{CFGM13} and
Canonne, Ron, and Servedio ~\cite{CRS} proposed a different model called
conditional sampling, which has emerged as a powerful tool for testing
properties of probability distributions. In this model, the testers are allowed to sample according to the
distribution conditioned on any specific subset of the domain. If the
distribution, $\mu$, is over the domain $\Sigma$, the
tester can submit any subset $S\subseteq \Sigma$ and receive a sample $i
\in S$ with probability $\mu(i)/\sum_{j\in S} \mu(j)$, where $\mu(i)$ is the probability of $i$ occurring 
when a sample is drawn from the distribution $\mu$. 
 
\cite{CFGM13,CRS} proved that in the conditional sampling model, 
testing uniformity, testing identity to a known distribution, and testing any label-invariant 
property of distributions is easier than with the unconditional sampling model.
 Specifically, one can get an algorithm for testing uniformity using
$\tilde{O} (1/\epsilon^{2})$ conditional samples (conditioning on arbitrary subsets
of size $2$)~\cite{CRS} .  Falahatgar et~al.~\cite{FalahatgarJOPS15}, improving an upper bound of $\OhT(1/\epsilon^4)$ in \cite{CRS}, showed that testing identity to a known distribution 
could also be done using $\OhT(1/\epsilon^2)$ conditional samples. 
They also showed that there exists an algorithm to test identity
between two unknown distributions on $\Sigma$ using $\OhT(\log \log |\Sigma|/\epsilon^5)$ conditional samples. In \cite{AcharyaCK15}, Acharya, Canonne, and Kamath showed a lower bound of $\Omega(\sqrt{\log \log n})$ for testing the equivalence of two unknown distributions. 

In the conditional sampling model, the sample complexity depends on the
structure of the condition, \emph{i.e.}, the structure of the subsets (of the domain) on which the 
distribution is conditioned for drawing samples. Naturally, if there is no restriction on
the condition, the tester can sample conditioned on arbitrary subsets, and the sample complexity improves. 
In \cite{CRS}, the authors presented an algorithm for testing whether a distribution over $\{1, \dots, n\}$ is uniform, 
with sample complexity $\tilde{\Theta} (1/\epsilon^{2})$ when conditioning on arbitrary subsets of size $2$. However, when the condition set was 
structured and restricted to intervals, they proved a lower bound of sample
complexity $\Omega\left(\frac{\log n}{\log \log n}\right)$. In \cite{ICALP:Canonne15}, Canonne showed that conditioning on interval improves the query complexity of monotonicity testing. 
Hence it is important to consider the plausible restrictions on the
conditions arising from the structure of the domain. 

 While \cite{CRS} studied some of the restrictions of the conditions, there are many more 
restrictions, which arise from the structure of the domain and/or arise from other applications,
which are yet to be studied. One such important case is when the domain is a Cartesian product of set and one is allowed to 
condition on the Cartesian product of subsets, but not on arbitrary subsets of the domain. 
 
\subsubsection*{Testing Joint Distributions: Subcube Conditioning}
%\textsc{Subcube Conditioning }
\label{sec:subcube-conditioning}

In practice, data are often multi-dimensional. In Cryptography, the keys are often defined over
$\bool^n$. Solutions to SAT formulae are over $\bool^n$ as well. On the other hand, the Lottery Tickets are defined over $[m]^n$ for 
some $m\in \N$  (each ticket contains $n$ numbers, each from the set $[m]$). Data analysts
often get data of million dimensions (features). With the higher dimension,
comes the ``curse of dimensionality.'' The sample complexity of the
testers is exponential in dimension \cite{AcharyaDK15,BFLKRW,DiakonikolasK16}, prohibiting practical
applications. Very recently, \cite{DaskalakisDK16} considered testing higher dimensional 
\emph{structured} distributions modelled using Markov Random Fields and achieved polynomial (in the dimension) 
sample complexity under the Ising model. \cite{CDKS16,DaskalakisP16} considered testing properties of \emph{structured} 
distributions using the probabilistic graphical model and achieved sublinear complexity for certain 
properties of Bayesian networks. However, all these results assume the distribution is structured and has certain properties. 
But for arbitrary distributions, testing with practical complexity remains a big concern.  

One can be hopeful that using conditional sampling, testing properties of \emph{arbitrary} joint distributions with practical 
complexity can be achieved. In that case, the assumptions are imposed on the sampling model. Finding a correct and natural sampling model is a 
challenge in itself. While joint distributions can also be viewed as a distribution over a larger domain, the marginals' domains may differ. Hence sampling conditioned on arbitrary subsets (as used in \cite{CRS,CFGM13}) may not be feasible in real life. 

 In \cite{CRS}, authors also considered structured conditioning, namely \textsc{Icond} (conditioning over an interval) and \textsc{PCond} (conditioning over a pair of points). 
\textsc{Icond} requires the domain to be well ordered. Moreover, for both cases, one should be able to sample from arbitrary intervals. For a joint 
distribution, the natural ordering of the domain is a pair; it involves ordering in the dimensions coupled with ordering 
in the \emph{individual} domains. For such an ordering, an arbitrary interval is required for the \textsc{Icond} 
tester need not be succinctly encodable and remains impractical.    

\subsection{Our Results}
\label{sec:our-results}

In this paper, we propose the subcube conditioning model and analyze property testing of joint distributions in that model.

Informally, the subcube conditioning model can be described in the
following way. Let $\Sigma^n$ be the domain of the distribution $\mu$. The Subcube Conditioning Oracle accepts  $A_1,A_2,\cdots,A_n\subseteq \Sigma$ and constructs $S=A_1\times A_2 \times\cdots\times A_n$  as the condition set. The oracle returns a vector $x=(x_1,x_2,\cdots,x_n)$, where each $x_i\in A_i$, with probability $\mu(x)/(\sum_{w\in S} \mu(w))$. If $\mu(S)=0$, we assume the oracle returns an element from $S$ uniformly at random. We will call these kinds of samples subcube-conditional-samples and the corresponding sample complexity subcube-conditional-sample complexity. There is no restriction on the individual $A_i$s. {\it They may be unstructured or structured as pairs or intervals as used in \cite{CRS, CFGM13}}.

\subsubsection*{Motivation of SubCube conditioning}
\label{sec:possible-application}
We believe the subcube conditional sampling model is mathematically interesting in itself. Every Boolean function can be 
modelled as a subgraph of a hypercube. Testing a property of a Boolean function translates to testing some property of the 
resulting subgraph. The conditional sampling model is equivalent to sampling over the edges of such subgraph, i.e., fixing 
some vertices, sampling over the edges, and checking the properties of the adjacent vertices. We argue sampling over the hypercube arises naturally in many areas.

\textbf{Database Query}. A typical ``SELECT'' query to a database often looks like \texttt{SELECT field1 WHERE field2= cond1 and field2 = cond2}. The response 
to such a query is all the tuples which satisfy \texttt{cond1} and \texttt{cond2}. Sampling over such tuples is indeed conditional sampling.

\textbf{Side Channel Cryptanalysis.} In modern Cryptography, schemes
are often ``proven'' secure (no efficient attack algorithm exists) under the assumption that the keys,
internal randomness, and internal memory are inaccessible to the
adversary. However, in practice, Cryptographic schemes are deployed in
a wide variety of devices, specifically hand-held devices and smart
cards. This situation leads to the ``side channel attacks'' where
tampering with the keys or internal randomness is feasible. Specifically, the cryptanalytic techniques of fault attacks fix/modify some bits and test the resulting
distributions. The subcube conditioning model captures this attack scenario (fixing some bits and testing on the resulting subcube).

Our results in this paper can be viewed as  proof that ``indistinguishability'' with uniform (in fact any known distribution) cannot be proven if an adversary can 
tamper with the internal state. \footnote{While this result is folklore in Cryptography, the subcube conditioning may be considered as the benchmark model while 
analyzing the efficiency of a fault attack.} 
 
\textbf{Verification of Random SAT solutions.} 
In software verification and related areas, random solutions to SAT
problems are often used as a backbone. However, testing whether the
solution that one algorithm generates is indeed uniform is a very
important problem. Unfortunately, the standard algorithms require
impractical complexity. Recently, Chakraborty et al. \cite{CM16} used the
conditional sampling model to get a practically deployable
solution. The model of subcube conditioning would be very effective to
this problem as one natural conditioning technique is to fix some variables of the SAT equation and then test the solution's distribution.

Recently \cite{GouleakisTZ17}, has significantly improved the runtimes of sublinear algorithms for k means clustering and weight estimation of minimum spanning tree using conditional samples. We believe the subcube conditioning can be used in this setting as well.  

We remark that the idea of subcube conditioning has also been mentioned in the literature related to property testing. In fact, analysis of joint distributions
using subcube conditioning was posed as a natural open problem in
\cite{CRS}.
\subsubsection*{Our Results}
\label{sec:upper-bound}

We focus on four fundamental properties of distributions:  given two joint distributions $\mu$ and $\mu'$ over $\Sigma^n$ we would like to test, using subcube-conditional-samples, if 
(a) $\mu$ is uniform, (b) $\mu$ is identical to $\mu'$ (when $\mu'$ is known in advance), (c) $\mu$ is identical to $\mu'$ (when $\mu'$ is not known in advance and has to be accessed using 
conditional samples), and (d) $\mu$ is a product distribution. We have the following four theorems:

\begin{theorem} (Informal)
  Let $\mu$ be a probability distribution over $\Sigma^n$. There exists an algorithm for testing if $\mu$ is uniform, using $\Tilde{\Oh} (n^2/\epsilon^{2})$ subcube-conditional-samples.\footnote{\label{note1} $\OhT$ hides a polynomial function of $\log n$ and $\log (1/\epsilon)$.  }  

\end{theorem}
 
\begin{theorem} (Informal)
\label{thm:joint-idenity}
 Let $\mu$ be a known probability distribution over the set $\Sigma^n$. Let $\mu'$ be an unknown distribution over $\Sigma^n$. There exists an algorithm to test identity of $\mu'$ with $\mu$ using $\OhT(n^2/\epsilon^{2})$ subcube-conditional-samples. \textsuperscript{\ref{note1}}
\end{theorem}

\begin{theorem} (Informal)
\label{thm:joint-unknown}
 Let $\mu,\mu'$ be unknown distributions over $\Sigma^n$. There exists an algorithm to test if $\mu'$ and $\mu$ are identical  using $\OhT(n^5\lld/\epsilon^{5})$ subcube-conditional-samples from both $\mu$ and $\mu'$. \textsuperscript{\ref{note1}}
\end{theorem}

\begin{theorem} (Informal)
\label{thm:indep-marg}
 Let $\mu$ be a probability distribution over the set $\Sigma^n$. There exists an algorithm to test whether $\mu$ is a product distribution using $\OhT(n^5\lld/\epsilon^{5})$ subcube-conditional-samples. \textsuperscript{\ref{note1}}
\end{theorem}

\subsubsection*{Comparison to Previous Results}

While conditional sampling has been studied in a number of articles in the recent past, and although subcube conditioning is a very natural model (that is also discussed in \cite{CRS}), 
as far as we understand, this is the first formal study on subcube conditioning. One of the main reasons for the lack of literature in this area is that the classical setting was not well 
studied either, 
till recently. Recently in \cite{CDKS16} Canonne et al. studied the problem of testing properties of joint distributions over the domain $\Sigma^n$. For example, for the fundamental problem of 
testing if 
the distribution is uniform, they observed that if the distribution is a product distribution (that is, the $n$ marginals are independent), then one needs 
$\Theta(\sqrt{n})$ samples. But if the distributions are not independent, then in the worst case, $\Theta(\Sigma^{n/2})$ samples are necessary. 

In comparison, we show that only $\OhT(n^2)$ subcube-conditional samples are necessary in the worst case, so we have an exponential improvement in the sample complexity.
Also, it is interesting to note that the sample complexity for uniformity testing in the subcube model is independent of $|\Sigma|$.  
This shows the power of subcube conditional samples and gets the query complexity to a more practical level. 
Also, from \cite{CDKS16} we know that $\Omega(\sqrt{n})$ conditional samples are necessary since, in the case of product distributions, conditional samples give no additional power 
over standard samples. 

A list of our results and comparison to previous results on standard sampling algorithms are given in Table~\ref{table1}.

\begin{table}[h!]
\centering
{\footnotesize
% \begin{center}
\begin{tabular}{ |c|c|c|c| } 
  \hline
  \multicolumn{1}{c}{Problems} &\multicolumn{2}{c}{Conditional Sampling}&\multicolumn{1}{c}{Traditional Sampling}\\
  \hline
  &Upper Bound [This paper] &Lower Bound & Upper and Lower Bound \\ 
 \hline\hline
 % Uniformity (Product distribution)  & $\OhT(n^2/\epsilon^{2})$ &$\Omega(\sqrt[4]{n})$ \\ 
 % \hline
Identity to the & ${ \OhT(n^2/\epsilon^{2})}$ &$\Omega(\sqrt{n}/\epsilon^{2})$ & $\Theta(|\Sigma|^{n/2}/\epsilon^{2})$ \\ 
uniform distribution & &\cite{CDKS16} &\cite{Paninski08} \\
\hline
Identity to a &${ \OhT(n^2/\epsilon^{2})}$&$\Omega(\sqrt{n}/\epsilon^{2})$& $\Theta(|\Sigma|^{n/2}/\epsilon^{2})$\\ 
known distribution & &\cite{{CDKS16}} &\cite{FOCS:ValVal14} \\
\hline
Identity between two &${\OhT(n^5\lld/\epsilon^{5})}$ & $\Omega\left(max\left(\sqrt{n}/\epsilon^{2},n^{3/4}/\epsilon\right)\right)$  & $\Theta\left(\mbox{max}(|\Sigma|^{2n/3}/\epsilon^{4/3}, |\Sigma|^{n/2}/\epsilon^{2})\right)$ \\
unknown distributions & &\cite{{CDKS16}} &\cite{conf/soda/ChanDVV14} \\
\hline
Identity to a   &  ${\OhT(n^5\lld/\epsilon^{5})}$ &$\Omega\left(max\left(\sqrt{n}/\epsilon^{2},n^{3/4}/\epsilon\right)\right)$& $\Theta(|\Sigma|^{n/2}/\epsilon^2)$\\
product distribution & &\cite{CDKS16} & \cite{AcharyaCK15,DiakonikolasK16}\\
 \hline
 \hline
\end{tabular}
}
\caption{comparison between sample complexity of testing joint distributions in Traditional Sampling Model and Subcube Conditioning Model.} \label{table1}
%\end{center}
\end{table}

\subsubsection*{Overview of Our Technique}
Let us start with the problem of testing if a given distribution is uniform. 
Let $\mu$ be a distribution over $\Sigma^n$ with marginals $\mu_1, \dots, \mu_n$.  

The simplest case is when $\mu$ is a product of $n$ independent distributions. That is, 
$\mu_i$'s are independent but not necessarily identical. But if $\mu$ is $\epsilon$-far from uniform 
, one expects to find at least one $\mu_i$ which is $\epsilon/n$-far from uniform. Then one can use any 
tester over $\Sigma$ if $\mu_i$ is far from uniform, which should make at most poly($n$) traditional queries.  
In fact, when $\mu$ is a product distribution over $\bool^n$, \cite{CDKS16} show that the uniformity and identity 
can be tested using $\Oh (\sqrt{n}/\epsilon^2)$ unconditional samples. As the marginals of $\mu$ are independent 
and over $\bool^n$, subcube-conditional-sampling is equivalent to unconditional sampling followed by projections, and 
hence subcube-conditional samples do not give any additional power in this setting.     

But if the $\mu_i$'s are not independent, then it is possible that all the individual marginals are uniform, but still, the $\mu$ is 
$\epsilon$-far from uniform. As has been observed in \cite{CDKS16}, any algorithm (using unconditional sampling) requires $\exp(n)$ queries. To circumvent this barrier, 
we need to use conditional samples. We define a notion of ``conditional distance". We show that there exists at least one $i \in [n]$ such that 
the expected ``conditional distance" of $i$th marginal from uniform is more than $\epsilon/\mbox{poly($n$)}$. 
Thus it is enough to test for all $i$ if the $i$th marginal is $\epsilon/\mbox{poly($n$)}$-far from uniform. 
We can use the testers from \cite{CRS,CFGM13} to test exactly that condition using poly($n$)
subcube-conditional samples. The central idea of the correctness of the algorithm is the
correct definition of the ``conditional distance" and the ``chain rule"  that proves that such an $i$ exists. Although the proof of the
``chain rule"  (given in Section~\ref{sec:tools}) is simple in
hindsight, it is a powerful tool that acts as the central backbone for all our upper-bound proofs. Moreover, it gives the flexibility of using an adaptive 
or non-adaptive tester over $\Sigma$.

\subsection{Organization of the paper}
\label{sec:organization-paper}
In Section \ref{sec:notat-prel}, we define the notion of conditional distance and
SubCube Conditioning. The chain rule is described in Section
\ref{sec:tools}. In Section~\ref{sec:gen-ub} we present the identity
testers and the derived uniformity tester. In Section~\ref{sec:unknown}, the tester for testing 
identity between two unknown distributions is presented. In Section~\ref{sec:marg}, the
tester for the independence of marginals is described. In Appendix~\ref{appendix} we present a 
lower bound of $n^{1/4}$ for testing identity to the uniform distribution. This lower bound was proved independently of 
\cite{CDKS16} and although our lower bound is weaker than their lower bound of $\sqrt{n}$, we feel that our techniques
can be of independent interest. 
 
\section{Notations and Preliminaries}
\label{sec:notat-prel}
If $S$ is a set, $|S|$ denotes the size of the set. If $x$ is a vector of length $n$, $x_i$ denotes the $i^{th}$ element of $x$. $x^{(i)}$ denotes the substring of first $i$ elements of $x$; $x^{(i)}= (x_1 ,x_2,\cdots,x_i)$. We denote the $n$-th harmonic number by $H(n)$. 

For any set $\Omega$, we denote by $\mathcal{U}_{\Omega}$ the uniform distribution with support $\Omega$. In most cases, the support of the distribution would be clear from the context and in that case, we would drop the subscript and use $\mathcal{U}$ as the uniform distribution over the support in question. 

If $\mu$ is a distribution with support $\Omega$, for any $x\in \Omega$, we will denote by $\Pr_{\mu}(x)$ the probability the $x$ occurs when a random sample is drawn from $\Omega$ according to $\mu$. If $\mu$ is a joint distribution, $\mu_i$ denotes the $i^{th}$ marginal distribution of $\mu$. 

If $\mu$ is a distribution over $\Sigma^n$ with the marginals $\mu_1, \dots, \mu_n$ and if the marginals are independent (that is, $\mu$ is a product distribution) then we would write $\mu = \mu_1\otimes \dots \otimes \mu_n$. 
%If all the $\mu_i$ are the same, say $\mu'$, then we will write $\mu$ as $\mu'^{\otimes n}$.

\

\textsc{Total Variation Distance.}
Let $\mu,\mu'$ be two distributions with support $\Omega$. The variation distance between $\mu$ and $\mu'$ denoted by $d(\mu, \mu')$ is defined as 
$$d(\mu,\mu') := \frac{1}{2}\sum_{x\in\Omega} \left|\Pr_{\mu}(x) - \Pr_{\mu'}(x)\right|.$$

We say $\mu$ and $\mu'$ are \emph{$\epsilon$-far} (or $\mu$ is $\epsilon$-far from $\mu'$), when
$d(\mu,\mu') \geq\epsilon.$

If $\mu$ is a distribution with support $\Omega$ and $A\subseteq \Omega$, then by $(\mu\mid A)$, we denote the distribution over the support $A$.  
For any $x\in A$, the probability that $x$ occurs when a random sample is drawn from $A$ (according to the distribution $(\mu\mid A)$)  is given by 
$$\Pr_{\mu \mid A}(x) = \frac{\Pr_{\mu}(x)}{\sum_{y\in A}\Pr_{\mu}(y)}.$$

\

\textsc{Hellinger Distance.}
Let $\mu,\mu'$ be two distributions with support $\Omega$. The Hellinger distance between $\mu$ and $\mu'$ denoted by $H(\mu, \mu')$ is defined as 
$$H(\mu, \mu') = \frac{1}{\sqrt{2}}\sqrt{\sum_{x\in \Omega}\left(\sqrt{\Pr_{\mu}(x)} - \sqrt{\Pr_{\mu'}(x)}\right)^2} = \sqrt{\left(1 - \sum_{x\in \Omega}\sqrt{\Pr_{\mu}(x)\Pr_{\mu'}(x)}\right)}$$

Hellinger distance has some nice properties and is useful for bounding lower and upper bounding variation distance. 

$$d(\mu, \mu') \leq 2H(\mu, \mu') \leq 2\sqrt{d(\mu, \mu')}$$

Also for any two product distributions $\mu = \mu_1 \otimes \dots \otimes \mu_n$ and $\mu'= \mu'_1\otimes \dots \otimes \mu'_n$ 
$$H(\mu, \mu')^2 \leq \sum_{i=1}^n H(\mu_i, \mu'_i)^2. $$

\

\textsc{Conditional Distance.}
Let $\mu,\mu'$ be two distributions over $\Omega$. Let $A\subseteq
\Omega$.  The variation distance between $\mu$ and $\mu'$ conditioned on $A$ (denote by $d(\mu,\mu'|A)$) is 
defined as $$d(\mu,\mu'|A) :=  \frac{1}{2}\sum_{x\in\Omega}\left| \Pr_{\mu\mid A}(x)- \Pr_{\mu'\mid A}(x)\right|.$$

We say $\mu$ and $\mu'$ are \emph{$\epsilon$-far}, conditioned on
$A$, when $d(\mu,\mu'|A) \geq\epsilon.$ 

\
\textsc{Subcube Conditioning.}
In this paper, we work with joint distributions; $\Omega=\Sigma^n$ for some set $\Sigma$. We consider conditional distance under the condition on $A=A_1\times A_2 \times\dots \times A_n$ where each $A_i\subseteq \Sigma$.
 
Let $\mu$ be a distribution over $\Sigma^n$ and $X=(X_1, X_2, \dots, X_n)$ be a random variable distributed according to $\mu$. $\mu^{(i)}$ denotes the distribution over $\Sigma^i$ where for every $x\in \Sigma^i$, $$\Pr_{\mu^{(i)}} (x) = \Pr_{X\sim \mu}[(X_1,X_2,\cdots,X_i)= (x_1,x_2,\cdots,x_i)] .$$ 

Let $w\in \Sigma^j$ for some $j < i$. $\mu_i\mid w$ denotes the marginal distribution $\mu_i$ when the first $j$ random variables are fixed to $w$. 

$$\Pr_{\mu_i\mid w} (x) = \Pr_{X\sim \mu}[X_i= x| \bigwedge_{k=1}^j X_k=w_k].$$

\begin{definition}
Let $\mu, \mu'$ be two distributions over $\Sigma^n$. The \emph{conditional marginal distance} of $\mu_i$ and $\mu_i$ conditioned on $w$ is given by 
$$d(\mu_i,\mu'_i\mid w) =  \frac{1}{2}\sum_{x\in \Sigma} \left| \Pr_{\mu_i \mid w} (x) - \Pr_{\mu'_i\mid w} (x) \right| $$

The \emph{average conditional distance} between $\mu_i$ and $\mu'_i$ is defined by

$$\Exp _{w \sim \mu^{(i-1)}}[ d(\mu_i,\mu'_i|w)]= \sum_{w\in
  \Sigma^{i-1}}\Pr_{\mu^{i-1}}(w) d(\mu_i,\mu'_i|w).$$  
\end{definition}

\subsubsection*{The SubCube Condition Model}
\label{sec:model}

  Let $\mu$ be a distribution over $\Sigma^n$. A \emph{subcube conditional oracle for $\mu$}, denoted
$\textsc{SubCond}_\mu$,
 takes as input a sequence of sets $\{A_i\}_{i \in [n]}$, $A_i \subseteq \Sigma$.  Let $A$ be the product set  $A_1\times \dots \times A_n$.
 The oracle returns an element $x \in \Sigma^n$
 with probability $ \frac{\Pr_{\mu}[x]}{\sum_{x\in A}\Pr_{\mu}[x]}$
 independently of all previous calls to the oracle.

An $(\epsilon,\delta)\mbox{-}\textsc{SubCond}$ tester for a property $\mathcal{P}$ with conditional sample complexity $t$
%$t:\reals \times \reals \times \naturals \times\naturals \to \naturals$ 
is a randomized algorithm, that receives $ 0<\epsilon,\delta<1 $, $ n \in \mathbb{N}$ and  oracle access to $\textsc{SubCond}_\mu$, and operates as follows.
\begin{enumerate} 
\item In every iteration, the algorithm (possibly adaptively) generates a set $A=A_1\times A_2\times\cdots\times A_n\subseteq \Sigma^n$, based on the transcript and its internal coin tosses, and calls the conditional oracle with $A$ to receive an element $x$, drawn according to the distribution $\mu$ conditioned on $A$.
\item Based on the received elements and its internal coin tosses, the algorithm accepts or rejects the distribution $\mu$.
\item The algorithm makes at most $t$ queries to $\textsc{SubCond}_\mu$, where $t$ can depend on $\epsilon, \delta, \Sigma$ and $n$.
\end{enumerate}

If $\mu$ satisfies $\mathcal{P}$, then the algorithm must accept with probability at least $1-\delta$, and if $\mu$ is $\epsilon$-far from all distributions satisfying $\mathcal{P}$, then the algorithm must reject with probability at least $1-\delta$.  

We will call such a tester an $(\epsilon,\delta)\mbox{-}\textsc{SubCond}$ $\mathcal{P}$-tester. For example an $(\epsilon,\delta)\mbox{-}\textsc{SubCond}$ Uniformity-tester is an $(\epsilon,\delta)\mbox{-}\textsc{SubCond}$ tester that tests if the given distribution is uniform, an $(\epsilon,\delta)\mbox{-}\textsc{SubCond}$ Identity-tester is an $(\epsilon,\delta)\mbox{-}\textsc{SubCond}$ tester that tests if the given distribution is identical to a known distribution and an $(\epsilon,\delta)\mbox{-}\textsc{SubCond}$ Product-tester is an $(\epsilon,\delta)\mbox{-}\textsc{SubCond}$ tester that tests if the given distribution is a product distribution or far from all the product distributions.

\section{Chain Rule of Conditional Distances}
\label{sec:tools}
Let $\mu$ and $\mu'$ be two distributions over $\Sigma^n$, and let $X=(X_1, X_2, \dots, X_n)$ and $X' = (X'_1, X'_2, \dots, X'_n)$ be the 
corresponding random variables.  For any $1\leq i\leq n$, we denote by $\mu_i$ and $\mu'_i$ the distributions of the $i$th marginals of $\mu$  and $\mu'$ respectively.

\

\begin{lemma}[Chain Rule of Conditional Distances]
  \label{lemma:chain}
Let $\mu$ and $\mu'$ be two distributions over $\Sigma^n$, and let $X=(X_1, X_2, \dots, X_n)$ and $X' = (X'_1, X'_2, \dots, X'_n)$ be 
two random variables with distribution $\mu$ and $\mu'$ respectively. Then the following holds.
\begin{align*}
  d(\mu,\mu')\leq d(\mu_1,\mu'_1)+ \sum_{i=2}^n \Exp _{w \sim \mu^{(i-1)}}[ d(\mu_i,\mu'_i|w)]  
\end{align*}
\end{lemma}
 
\begin{proof}[Proof of Lemma~\ref{lemma:chain}:]
Let $w = (w_1,w_2,\dots,w_n)\in \Sigma^n$.

\noindent Let $2 \leq i \leq n$. Recall that $w^{(i)}$ denotes the substring of first $i$ elements of $w$.
\begin{eqnarray*}
2d(\mu^{(i)},\mu'^{(i)})&=&\sum_{w\in \Sigma^i} |\Pr_{\mu^{(i)}} (w)- \Pr_{\mu'^{(i)}} (w)|\\
   &=& \sum_{w\in \Sigma^i} |\Pr_{X\sim \mu} [\wedge_{j=1}^{i-1} X_j = w_j]\Pr_{X\sim \mu} [X_i=w_i|\wedge_{j=1}^{i-1} X_j = w_j] \\ &&\qquad\qquad \qquad\qquad  -\Pr_{X'\sim \mu'} [\wedge_{j=1}^{i-1} X_j' = w_j]\Pr_{X'\sim \mu'} [X_i'=w_i|\wedge_{j=1}^{i-1} X_j' = w_j]|\\
&\leq & \sum_{w\in \Sigma^i} \left|\Pr_{X\sim \mu} [\wedge_{j=1}^{i-1} X_j = w_j]\left(\Pr_{X\sim \mu} [X_i=w_i|\wedge_{j=1}^{i-1} X_j = w_j]- \Pr[X_i'=w_i|\wedge_{j=1}^{i-1} X_j' = w_j]\right)\right|\\  &&+ \sum_{w\in \Sigma^i} \left| \Pr[X_i'=w_i|\wedge_{j=1}^{i-1} X_j' = w_j]\left(\Pr_{X\sim \mu} [\wedge_{j=1}^{i-1} X_j = w_j]- \Pr_{X'\sim \mu'} [\wedge_{j=1}^{i-1} X_j' = w_j] \right)\right|
\end{eqnarray*}

Now, the second term reduces to,
\begin{eqnarray*}
  &\sum_{w\in \Sigma^i}& \left| \Pr[X_i'=w_i|\wedge_{j=1}^{i-1} X_j' = w_j]\left(\Pr_{X\sim \mu} [\wedge_{j=1}^{i-1} X_j = w_j]- \Pr_{X'\sim \mu'} [\wedge_{j=1}^{i-1} X_j' = w_j] \right)\right|\\
&=& \sum_{w'\in \Sigma^{i-1}}\left|\Pr_{X\sim \mu} [\wedge_{j=1}^{i-1} X_j = w'_j]- \Pr_{X'\sim \mu'} [\wedge_{j=1}^{i-1} X_j' = w'_j] \right| \sum_{w_i\in \Sigma}\Pr_{X'\sim\mu'}[X_i'=w_i|\wedge_{j=1}^{i-1} X_j' = w'_j]\\
&=&\sum_{w'\in \Sigma^{i-1}}\left|\Pr_{X\sim \mu} [\wedge_{j=1}^{i-1} X_j = w'_j]- \Pr_{X'\sim \mu'} [\wedge_{j=1}^{i-1} X_j' = w'_j]\right|\\
&=& \sum_{w'\in \Sigma^{i-1}} |\Pr_{\mu^{(i-1)}} (w')- \Pr_{\mu'^{(i-1)}} (w')|\\ &=& 2 d(\mu^{(i-1)},\mu'^{(i-1)}).
\end{eqnarray*}
The second equality follows from the fact that for each $w'\in \Sigma^{i-i}$,
$\sum_{w_i\in \Sigma}\Pr[X_i'=w_i|\wedge_{j=1}^{i-1} X_j' = w'_j]=1.$\footnote{If $w'$ is outside of the support of $\mu'$, like in \cite{CFGM13}, we can define the conditional probability to be uniform over $\Sigma$}
Hence,
\begin{eqnarray*}
   d(\mu^{(i)},\mu'^{(i)}) \leq  d(\mu^{(i-1)},\mu'^{(i-1)})+ \sum_{w \in \Sigma^{i-1}} \Pr_{\mu^{(i-1)}}(w) d(\mu_i,\mu'_i|w)
\end{eqnarray*}

Solving the recursion, we get the lemma.
\end{proof}

\
Arranging the marginals by the increasing order of the average conditional
distance, we get the immediate corollary.
\begin{lemma}
\label{lemma:count}
If $d(\mu,\mu') \geq \epsilon$, then there exists a $c\leq \lceil \log n\rceil$ such that
$$2^{c-1} \leq \left|\left\{ i\in [n] \mid \Exp _{w \sim \mu^{(i-1)}}[ d(\mu_i,\mu'_i|w)] \geq \frac{\epsilon}{2^c\Hn}\right\}\right|$$ 
\end{lemma}

% We refer the reader to Appendix \ref{sec:proof-lemma-count} for the proof.

\begin{proof}[Proof of Lemma~\ref{lemma:count}]
Without loss of generality let $i_1,i_2,\dots,i_n $ be indices such that
$$ \Exp _{w \sim \mu^{(i_1-1)}}[ d(\mu_{i_1},\mu'_{i_1}|w)]\geq
\Exp _{w \sim \mu^{(i_2-1)}}[ d(\mu_{i_2},\mu'_{i_2}|w)] \geq \Exp _{w \sim \mu^{(i_n-1)}}[ d(\mu_{i_n},\mu'_{i_n}|w)] $$

We will need the following claim. 

 \begin{claim}
\label{claim:index}
 There exists $k\in [n]$ such that  
$$ \Exp _{w \sim \mu^{(i_k-1)}}[ d(\mu_{i_k},\mu'_{i_k}|w)]\geq \epsilon/(k \Hn)$$
  \end{claim}
  
Let $k$ be the index from Claim~\ref{claim:index}. We put $c=\lceil
\log k\rceil$ to get
$\epsilon/ 2^c \Hn \leq \epsilon/k \Hn $. Clearly
$$ \left|\left\{ i\in [n] \mid \Exp _{w \sim \mu^{(i-1)}}[ d(\mu_{i},\mu'_{i}|w)] \geq
    \frac{\epsilon}{2^c\Hn}\right\}\right|
        \geq k \geq 2^{c-1}.$$
\end{proof}

 \begin{proof}[Proof of Claim~\ref{claim:index}]
    If no such $k$ exists, then 
   $$  d(\mu,\mu') \leq \sum_ {k=1}^n  \Exp _{w \sim \mu^{(i_k-1)}}[ d(\mu_{i_k},\mu'_{i_k}|w)] < \sum_{k=1}^n   \epsilon/(k \Hn)\leq \epsilon$$
   which contradicts the distance assumption in Lemma \ref{lemma:count}.
  \end{proof}

\section{Testing Identity with a known distribution}
\label{sec:gen-ub}

In this section, we present an identity tester of Sample complexity
$\OhT(n^2/\epsilon^{2})$. We recall the following result proved in \cite{FalahatgarJOPS15}.

\begin{lemma}\cite{FalahatgarJOPS15}\label{lem:testidentitysingle}
Let $\mu$ be a known distribution over $\Sigma$. Given $0<\epsilon <1$ and  $0<\delta<1$ and a distribution $\mu'$ over $\Sigma$
there is an adaptive $(\epsilon, \delta)$-\textsc{SubCond} Identity Tester with conditional sample complexity $\OhT (\frac{1}{\epsilon^2}\log(\frac{1}{\delta}))$. In other words, 
there is a tester that  draws $\OhT (\frac{1}{\epsilon^2}\log(\frac{1}{\delta}))$ conditional samples and 
\begin{itemize}
\item if $\mu=\mu'$,  then the tester will accept with probability $(1 - \delta)$, and
\item if $d(\mu, \mu')\geq \epsilon$ then the tester will reject with probability $(1-\delta)$.
\end{itemize}
\end{lemma}

 Let $\mu$ be a known distribution over $\Sigma^n$, $\mu'$ be an unknown distribution over $\Sigma^n$ that can be accessed via $\textsc{SubCond}_{\mu'}$ oracle, and $\epsilon$ be the target distance. The following algorithm tests the identity of $\mu'$ with $\mu$. We use the identity tester \textsf{BasicIDTester} over $\Sigma$ guaranteed by Lemma~\ref{lem:testidentitysingle} as a subroutine.

\begin{algorithm}[H]
\caption{The Identity Tester for Joint Distributions}
\label{alg:gen-ub}
\begin{algorithmic}[1]
\STATE $\delta=1/3$.
\STATE \label{step:delta} $\delta'= \delta\epsilon/64n(\log n)^2$
\FOR {$j = 1$ to $\log n +1$}
\STATE \label{step:eps} $\epsilon_j= \epsilon/2^{j}\Hn $
\STATE $\ell_j= \log\left( \frac{2^{j+1}\Hn}{\epsilon}\right)$
 \STATE Create a set $S_j$ by sampling, with replacement, $(4n/2^j)$ element from $[n]$ uniformly at random.
 \FORALL {$i \in S_{j}$}
 
 \FOR{$k=0$ to $\ell_j$}
 \STATE \label{step:epsp} $\epsilon_{(j,k)}=2^{k-1}\epsilon_j$
 \STATE \label{step:deltap}$\delta_k=\delta'/(k+3)^2$
    \FOR{$t=1$ to $ 2^{k+2}(k+3)^2$}
   \STATE \label{step:sample} Sample $w \sim \mu$. Let $w =(w_1,\cdots, w_n)$. 
    \STATE  Consider the distribution $\mu_i  \mid w^{(i-1)}$.  
    \STATE \label{step:test} If $\textsf{BasicIDTester}(\mu_i|w^{(i-1)}, \mu'_i | w^{(i-1)}, \epsilon_{(j,k)},\delta_k)$ rejects, Output REJECT
\ENDFOR
    \ENDFOR
\ENDFOR
\ENDFOR
\STATE Output ACCEPT
\end{algorithmic}
\end{algorithm}

\begin{theorem}
\label{thm:gen-ub}
Given any $0<\epsilon<1$, Algorithm~\ref{alg:gen-ub} is an $(\epsilon, \frac{1}{3})$ -\textsc{SubCond} Identity
 Tester for joint distributions with conditional sample complexity of
 $\OhT(n^2/\epsilon^2)$ where $\OhT$ hides a polynomial function of $\log n,\log \frac{1}{\epsilon}$.
\end{theorem}

\begin{note}
  \label{note:algoub}
  For any $0<\epsilon,delta<1$, one can obtain an $(\epsilon, \delta)$ -\textsc{SubCond} Identity Tester by standard techniques of error reduction. The query complexity would increase by a factor of $\log (1/\delta)$.
\end{note}

\subsection{Proof of Theorem \ref{thm:gen-ub}}
\label{sec:proof-theor-thm:genub}

Fix $\delta=\frac{1}{3}$. In Algorithm \ref{alg:gen-ub}, Step \ref{step:test} queries \textsf{BasicIDTester}. \textsf{BasicIDTester} needs conditional samples for testing
whether $d(\mu_i, \mu'_i \mid w^{(i-1)}) \geq \epsilon_{(j,k)}$. % are
% obtained by drawing samples from $\mu'$ conditioned on sets of the form
% $A = (A_1 \times \cdots \times A_n) \subseteq \Sigma^n$ .
To answer a conditional query with condition $B
\subseteq \Sigma$ for
the distribution $\mu'_i|w^{(i-1)}$, we set $A_j=\{w_j\}$ for
$j=1,2,\dots,i-1$, $A_i=B$, and $A_j=\Sigma$ for $j=i+1,\dots,n$, and query the \textsc{SubCond} oracle with the condition $A$. This correctly simulates the conditional oracle required by the underlying identity tester. Thus Algorithm~\ref{alg:gen-ub} is a \textsc{SubCond} Tester. 

\subsubsection{Sample Complexity of Algorithm~\ref{alg:gen-ub}}
By Lemma \ref{lem:testidentitysingle}, a  query to $\textsf{BasicIDTester}(\mu_i|w^{(i-1)}, \mu'_i | w^{(i-1)}, \epsilon_{(j,k)},\delta_k)$ requires
$\OhT({1}/{\epsilon_{(j,k)}^{2}})$ samples. Here $\OhT$ hides polylogarithmic factors of $|\Sigma|,\epsilon_{(j,k)}$ including the factors due to $\log (1/\delta_k)$.

For each index in $S_j$, the sample complexity is
\begin{align*}
  \sum_{k=0}^{\ell_j} \OhT\left(\frac{1}{\epsilon_{(j,k)}^{2}}\right) \Oh\left(2^kk^2\right)
& = \sum_{k=0}^{\ell_j} \OhT\left(\frac{2^kk^2}{2^{2k}\epsilon_j^2}\right) \qquad [\because \epsilon_{(j,k)}= 2^{k-1}\epsilon_j~\mbox{by step \ref{step:epsp}}]\\
& = \sum_{k=0}^{\ell_j} \OhT\left(\frac{k^2}{2^{k}\epsilon_j^2}\right) 
\end{align*}
Here $\OhT$ hides some polylogarithmic function of $k$ and $1/\epsilon_j$.  
As $k\leq \ell_j= \log\left( \frac{2}{\epsilon_j}\right)$, the expression can be bounded as
\begin{align*}
  \sum_{k=0}^{\ell_j} \OhT\left(\frac{k^2}{2^{k}\epsilon_j^2}\right)= \OhT\left(\frac{1}{\epsilon_j^2}\right)\sum_{k=0}^{\ell_j} \Oh\left(\frac{k^2}{2^{k}}\right) =\OhT\left(\frac{1}{\epsilon_j^2}\right)
\end{align*}
The last equality holds true as $\sum_{k\geq 0}\frac{k^2}{2^{k}}=6$.

The size of $S_j$ is $\frac{4n}{2^j}$. Adding over all possible $j$, we get the total sample complexity 
\begin{eqnarray*}
 \sum_{j=1}^{\log n +1} \frac{4n}{2^j} \OhT\left(\frac{1}{\epsilon_j^2}\right)
  &=&  \sum_{j=1}^{\log n +1} \frac{4n}{2^j} \OhT\left(\frac{2^{2j}H(n)^2}{\epsilon^2}\right)\qquad [\because \epsilon_j= \epsilon/2^j \Hn~\mbox{by step \ref{step:eps}}]\\
  &=&\OhT\left(\frac{nH(n)^2}{\epsilon^2}\right)  \sum_{j=1}^{\log n +1} 2^j 
~~= ~~\OhT(n^2/\epsilon^2)
\end{eqnarray*}

\subsubsection{Correctness of the Algorithm~\ref{alg:gen-ub}}
\label{sec:proof-correctness-algogenub}
\textsc{Completeness}. We will show that if $d(\mu,\mu')=0$,
  the algorithm will reject with probability at most $\delta$. 

Algorithm \ref{alg:gen-ub}, rejects $\mu'$ if there exists $i \in [n]$
and a sampled $w=(w_1, \cdots, w_n) \in \Sigma^n$ the underlying Identity Tester
rejects in the Step ~\ref{step:test}. 

 Suppose $\mu$ and $\mu'$ are identical. Then for all $w\in \Sigma^{i-1}$, $\mu_i| w$ is identical to $\mu'_i\mid w$.  For each query, \textsf{BasicIDTester} will reject in Step~\ref{step:test} with probability at most $\delta_k$. By union bound, the probability that the algorithm will reject $\mu'$ is at most 

 \begin{align*}
  \sum_{j=1}^{\log n +1}\frac{4n}{2^j} \sum_{k=0}^{\ell_j} (2^{k+2}(k+3)^2 \delta_k)
  &= \sum_{j=1}^{\log n +1}\frac{4n}{2^j} \sum_{k=0}^{\ell_j} (2^{k+2} \delta')\qquad [\because \delta_k = \delta'/ (k+3)^2~ \mbox{by step \ref{step:deltap}}]\\
  &= 16\delta'  \sum_{j=1}^{\log n +1}\frac{n}{2^j} \sum_{k=0}^{\ell_j} 2^{k}\\
  &< 16 \delta'  \sum_{j=1}^{\log n +1}\frac{n}{2^j} 2^{\ell_j+1}\\
  &= 64 \delta'  \sum_{j=1}^{\log n +1}\frac{n}{2^j} \frac{2^j H(n)}{\epsilon}\\
  &<  \frac{64 \delta' n (\log n)^2}{\epsilon}= \delta  \qquad [\because \delta' = \delta\epsilon/64n(\log n)^2~ \mbox{ by step \ref{step:delta}}]
\end{align*}

\textsc{Soundness.} Now, we prove the soundness of the Algorithm
\ref{alg:gen-ub}. Let $\mu$ be a distribution over $\Sigma^n$ and $d(\mu,\mu') \geq
\epsilon$. We shall show that Algorithm~\ref{alg:gen-ub} rejects
$\mu'$ with a probability of at least $2/3$.

\noindent Let 
\begin{equation*}
  \tau_c \eqdef \left\{ i\in [n] \mid \Exp _{w \sim \mu^{(i-1)}}[
    d(\mu_i,\mu'_i|w)] \geq \frac{\epsilon}{2^c\Hn}\right\}
\end{equation*}

\noindent Let $c \leq \ceil{\log n}$ be the integer guaranteed by Lemma~\ref{lemma:count}, such that $|\tau_c| \geq 2^{c-1}$. Note, $\ell_c=\lceil \log \left(\frac{2^{c+1}H(n)}{\epsilon}\right)\rceil$.  For each $i\in \tau_c$, for each $k\in [\ell_c] \cup \{0\}$ define 
  $$ \Gamma_{i,k}\eqdef \left\{ w\in\Sigma^{i-1} \mid
   d(\mu_i,\mu'_i|\wedge_{j=1}^{i-1}X_j = w_j)<
   \frac{2^{k-1}\epsilon}{2^{c}\Hn}\right\}$$

We require the following lemma based on Levin's economical work investment strategy~\cite{Goldreich17}.
 
%%%%%%%%%%%%%%%%%%%%%%%%%%%%%%%%%%%%%%%%%%%%%%%%%
%%%%% Levin's Strategy%%%%%%%%%%

\begin{lemma}
  \label{lemma:findingwnew}
Let $\mu$ be a distribution over $\Sigma^n$, and $\mu$ is
$\epsilon$-far from uniform. Let $X=(X_1, \cdots, X_n)$ be a random
variable with distribution $\mu$. Let $w=(w_1,w_2,\cdots, w_n)$ be a
random sample drawn from $\Sigma^n$ according to the distribution
$\mu$. Let $\epsilon_c=\frac{\epsilon}{2^{c}\Hn}$ and $\ell_c=\lceil \log \left(\frac{2}{\epsilon_c}\right)\rceil$.

Then for all $i\in \tau_c$, there exists $k\in [\ell_c] \cup \{0\}$,  

\begin{equation}
  \label{eq:1}
  \Pr_{w\sim \mu}\left[d(\mu_i,\mu'_i\mid w^{i-1})\geq 2^{k-1}\epsilon_c\right] \geq \frac{1}{2^{k}(k+3)^2}
\end{equation}

\end{lemma}

\begin{proof}[( Proof of Lemma \ref{lemma:findingwnew}.)]

  \noindent From Lemma~\ref{lemma:count}, for all index $i\in \tau_c$ 
$$ \Exp _{w \sim \mu^{(i-1)}}[ d(\mu_i,\mu'_i|w)]= \sum_{w\in
  \Sigma^{i-1}}\Pr_{\mu^{i-1}}(w) d(\mu_i,\mu'_i|w)\geq \frac{\epsilon}{2^{c}\Hn}$$
\noindent Fix $i\in \tau_c$.  Let us define
  \begin{align*}
    B_k&\eqdef \{w\in \Sigma^{i-1}: 2^{k-1}\epsilon_c\leq d(\mu_i,\mu'_i\mid w) < 2^{k} \epsilon_c \}\qquad k\in [\ell_c]\cup \{0\}\\
    B_{-1}&\eqdef \{w\in \Sigma^{i-1}: d(\mu_i,\mu'_i\mid w) < \epsilon_c/2 \}
  \end{align*}
By construction, $B_{\ell_c+1}=\emptyset$.
We shall prove that there exists $k\in [\ell_c]\cup\{0\}$ such that $\Pr_{w\sim \mu}[w \in B_k] \geq \frac{1}{2^{k}(k+3)^2}$. Suppose, towards contradiction, for all $k\in [\ell_c]\cup\{0\}$,  $\Pr_{w\sim \mu}[w \in B_k] < \frac{1}{2^{k}(k+3)^2}$. Then   
    \begin{align*}
      \Exp_{w \sim \mu^{(i-1)}}[ d(\mu_i,\mu'_i|w)]
      &= \sum_{w\in\Sigma^{i-1}}\Pr_{\mu^{i-1}}(w) d(\mu_i,\mu'_i|w)\\
      &= \sum_{w\in B_{-1}} \Pr_{\mu^{i-1}}(w) d(\mu_i,\mu'_i\mid w) + \sum_{k\in[\ell_c]\cup\{0\}}\sum_{w\in B_k} \Pr_{\mu^{i-1}}(w) d(\mu_i,\mu'_i\mid w)\\
      &< \Pr_{w\sim \mu}[w \in B_{-1}]\frac{\epsilon_c}{2}  + \sum_{k\in[\ell_c]\cup\{0\}} \Pr_{w\sim \mu}[w \in B_k] 2^k \epsilon_c\\
      &< \frac{\epsilon_c}{2}+ \sum_{k\in[\ell_c] \cup\{0\}} \frac{1}{2^{k}(k+3)^2}2^k \epsilon_c\\
      &= \frac{\epsilon_c}{2}+ \sum_{k\in[\ell_c]} \frac{\epsilon_c}{(k+2)^2}\\
      &< \epsilon_c
    \end{align*}
In the last inequality we used the fact that $\sum_{k\in [\ell_c]}\frac{1}{(k+2)^2}< \sum_{k\geq 0}\frac{1}{(k+2)^2}$ which is less than $1/2$.

\end{proof}

% \noindent({\small Proof of soundness continued.})\\
\noindent By Lemma \ref{lemma:findingwnew}, there exists $0\leq k\leq \ell_c$, such that,
 \begin{equation}
   \label{eq:2}
   Pr_{w\sim \mu^{i-1}}[w\in \Gamma_{i,k}] < \left(1- \frac{1}{2^{k}(k+3)^2}\right) 
 \end{equation}

Let $S_j$ be the set of indices sampled in Step 3 in the $j^{th}$
iteration. If Algorithm~\ref{alg:gen-ub} fails to reject $\mu'$, one of the
following three cases happens.
\begin{enumerate}
\item No index from $\tau_c$ was sampled in $S_j$. Specifically,
  $S_c\cap \tau_c = \emptyset$. The probability of this event is  
$$\left(1- \frac{|\tau_c|} {n}      \right)^{4n/2^c } \le e^{-2}.$$  

\item For all index $i\in S_c \cap \tau_c$, for each $k\in [\ell_c] \cup \{0\}$,  all the sampled $w$'s are
  from the set $\Gamma_{i,k}$. The probability of this event
  is  $$\left(1-\frac{1}{2^{k}(k+3)^2}\right)^{2^{k+2}(k+3)^2}
  \le e^{-4}.$$

\item For all index $i\in S_c \cap \tau_c$, for each $k\in [\ell_c] \cup \{0\}$, for all the sampled $w
  \notin \Gamma_{i,k}$, underlying identity tester fails to
  reject. The probability of such an event is at most
  $\delta'$, which is less than $1/100$ for
  $n\geq 2$.
  
\end{enumerate}

Hence, the probability that Algorithm~\ref{alg:gen-ub} fails to reject
$\mu'$ is at most $e^{-2}+e^{-4}+1/100 < 1/3$.

This completes the proof of Theorem \ref{thm:gen-ub}. \qed

\subsection{Uniformity Tester for Arbitrary Joint Distribution}
% \label{sec:distr-with-low}
If we set $\mu$ to be the uniform distribution, then Algorithm~\ref{alg:gen-ub} gives us a Uniformity Tester. Hence, we get the following as a corollary of Theorem~\ref{thm:gen-ub}.

\begin{theorem}
  \label{thm:uniform}
  Given any $0<\epsilon<1$, there exists  an $(\epsilon,\frac{1}{3})$-\textsc{SubCond} Uniformity
 Tester for any joint distribution with conditional sample complexity of
 $\OhT(n^2/\epsilon^2)$ where $\OhT$ hides a polynomial function of $\log n,\log \frac{1}{\epsilon}$.
\end{theorem}

\section{Identity Testing between Unknown Joint Distributions}
\label{sec:unknown}
In this section, we present Algorithm \ref{alg:unknown-ub} to test identity when both $\mu$ and $\mu'$ are unknown. The first change, from Algorithm \ref{alg:gen-ub}, we need to make is in Step \ref{step:sample}. In this case, we can no longer sample on our own. However, we can query $\mu$ to get $w$. Secondly, instead of Algorithm {\sf BasicIDTester}, we need to use Algorithm {\sf BasicUnknown} guaranteed by the following lemma.  

\begin{lemma}{\cite{FalahatgarJOPS15}}\label{lem:testunknwonsingle}
Given $0<\epsilon <1$ and  $0<\delta<1$ and distributions $\mu,\mu'$ over $\Sigma$
there is an $(\epsilon, \delta)$-Identity Tester with conditional sample complexity $\OhT (\frac{\lld}{\epsilon^5}\log(\frac{1}{\delta}))$. In other words, 
there is a tester that  draws $\OhT (\frac{\lld}{\epsilon^5}\log(\frac{1}{\delta}))$ independent conditional samples and 
\begin{itemize}
\item if $\mu=\mu'$,  then the tester will accept with probability $(1 - \delta)$, and
\item if $d(\mu, \mu')\geq \epsilon$ then the tester will reject with probability $(1-\delta)$.
\end{itemize}
\end{lemma}

\begin{algorithm}
\caption{The Identity Tester for two Unknown Joint Distributions}
\label{alg:unknown-ub}
\begin{algorithmic}[1]
\STATE $\delta=1/3$.
\STATE $\delta'= \delta\epsilon/64n(\log n)^2$
\FOR {$j = 1$ to $\log n +1$}
\STATE $\epsilon_j= \epsilon/2^{j}\Hn $
\STATE $\ell_j= \log\left( \frac{2}{\epsilon_j}\right)$
 \STATE Create a set $S_j$ by sampling, with replacement, $(4n/2^j)$ element from $[n]$ uniformly at random.
\FORALL {$i \in S_{j}$} 
  \FOR{$k=0$ to $\ell_j$}
 \STATE $\epsilon_{(j,k)}=2^{k-1}\epsilon_j$
 \STATE $\delta_k=\delta'/(k+3)^2$
   \FOR{$t=1$ to $ 2^{k+2}(k+3)^2$}
   \STATE \label{step:sample} Query oracle $\mu$ to get $w \sim \mu$. Let $w =(w_1,\cdots, w_n)$. 
    \STATE  Consider the distribution $\mu_i  \mid w^{(i-1)}$.  
    \STATE \label{step:test2} If $\textsf{BasicUnknown}(\mu_i|w^{(i-1)}, \mu'_i | w^{(i-1)}, \epsilon_{(j,k)},\delta_k)$ rejects, Output REJECT
    \ENDFOR
    \ENDFOR
    
\ENDFOR
\ENDFOR
\STATE Output ACCEPT
\end{algorithmic}
\end{algorithm}

To prove the correctness of Algorithm \ref{alg:unknown-ub}, we note that, in the chain rule, the expectation is over only one distribution. Hence it is sufficient to (unconditionally) query only $\mu$ to get $w$, and apply Lemma \ref{lemma:count}. The rest of the proof is exactly the same as in Section \ref{sec:gen-ub}.

\subsubsection*{Sample Complexity of Algorithm~\ref{alg:unknown-ub}}
\label{sec:sample-compl-ind-marginal}
By Lemma \ref{lem:testunknwonsingle}, each invocation of {\sf BasicUnknown}
with parameter $\epsilon_k$,$\delta_k$ requires
$\OhT(\log\log|\Sigma|/\epsilon_k^5)$ samples. As in the case for Algorithm \ref{alg:gen-ub}, for each index in $S_j$, the sample complexity is $\OhT (\lld/\epsilon^{5})$.
 Hence, the total sample
complexity of Algorithm \ref{alg:unknown-ub} is 
\begin{eqnarray*}
 \sum_{j=1}^{\log n +1}
  \frac{4n}{2^j}\times \OhT\left(\frac{\lld}{\epsilon_j^{5}}\right)
  &=& \sum_{j=1}^{\log n +1}
  \frac{4n}{2^j}\times \OhT\left(\frac{2^{5j}H(n)^5\lld}{\epsilon^{5}}\right)\\
&=& \OhT\left(\frac{n \Hn^5\lld}{\epsilon^5}\right)\sum_{j=1}^{\log n+1} 2^{4j}\\ 
&=& \OhT\left(\frac{n^5\lld}{\epsilon^5}\right)
\end{eqnarray*}

\begin{theorem}
  \label{thm:unknown}
  Given $0<\epsilon<1$, Algorithm \ref{alg:unknown-ub} is an $(\epsilon,\frac{1}{3})$-\textsc{SubCond} Identity
 Tester for two unknown joint distributions with conditional sample complexity of
 $\OhT\left(\frac{n^5\lld}{\epsilon^5}\right)$ where $\OhT$ hides a polynomial function of $\log n,\log \frac{1}{\epsilon}$.
\end{theorem}

\section{Testing Independence of Marginals}
\label{sec:marg}

Let $\mu$ be a probability distribution over $\Sigma^n$. In this
section, we present an algorithm to test whether $\mu$ is a product
distribution; i.e., \emph{whether all the marginals of $\mu$ are
independent or $\mu$ is far from all the product distributions}. 

Define $\mu'$ to be the product of marginals of $\mu$.
\begin{equation*}
  \Pr_{\mu'}(w)= \prod_{i=1}^n \Pr_{\mu_i}(w_i)~\forall~w\in\Sigma^n 
\end{equation*}

By definition, the marginal distributions $\mu'_i$ are exactly the
marginal distributions $\mu_i$. If $\mu$ is $\epsilon$-far from all the
product distributions, it is $\epsilon$-far from $\mu'$. Using the chain rule (Lemma
\ref{lemma:chain}), 
\begin{eqnarray*}
  d(\mu,\mu')& \leq &d(\mu_1,\mu'_1)+ \sum_{i=2}^n \Exp _{w \sim
    \mu^{(i-1)}}[ d(\mu_i,\mu'_i|w)]\\
     &=& \sum_{i=2}^n \sum_{w\in \Sigma^{i-1}} Pr_{\mu^{(i-1)}}(w)
     \left(\sum_{w_i\in \Sigma} \left| \Pr_{\mu_i}(w_i|w)-
         \Pr_{\mu'_i}(w_i|w)    \right|\right)\\
     % &=& \sum_{i=2}^n \sum_{w\in \Sigma^{i-1}} Pr_{\mu^{(i-1)}}[w]
     % \left(\sum_{w_i\in \Sigma} \left| \Pr_{\mu_i}(w_i|w)-
     %     \Pr_{\mu_i}(w_i)    \right|\right)\\
     &=& \sum_{i=2}^n \sum_{w\in \Sigma^{i-1}} Pr_{\mu^{(i-1)}}(w)
     d(\mu_i|w, \mu_i)    
\end{eqnarray*}

Therefore, we need to test whether there exists $i\in [n]$, such that
 the marginal distribution $\mu_i$ is far (on average) from the
 conditional marginal distribution $\mu_i|w$. As both $\mu_i$
 and $\mu_i|w$ is distributed over $\Sigma$, we can again use {\sf BasicUnknown} tester from
 \cite{FalahatgarJOPS15}, where identity between two unknown
 distributions is tested using $\OhT(\log \log |\Sigma|/\epsilon^5)$
 sample complexity. The only thing left is to sample $w$ according to
 $\mu^{i-1}$. Such a $w$ can be sampled by taking an unconditionally sampled string and selecting the first $i-1$ bit of that string. The rest of the algorithm is exactly the same as in Algorithm \ref{alg:unknown-ub}.   

\begin{theorem}
\label{thm:marg-ub}
For any $0<\epsilon<1$, there exists an $(\epsilon, \frac{1}{3})$-
 \textsc{SubCond} Product Tester for joint distributions with conditional sample complexity of
 $\OhT\left(\frac{n^5\lld}{\epsilon^{5}}\right)$, where $\OhT$ hides a polynomial function of $\log n, \log\left(\frac{1}{\epsilon}\right)$
\end{theorem}

The proof of Theorem \ref{thm:marg-ub} follows directly from Theorem \ref{thm:unknown}, and the observation that in this particular case, the (conditional) samples for $\mu_i$  can be produced by conditioning only on the $i^{th}$ index of $\Sigma^n$. % For completeness, we include detailed proof in Appendix. 

\section{Conclusion}
\label{sec:applications}
In this paper, we analyzed property testing of joint distributions in
the conditional sampling model. We considered the natural subcube
conditioning and presented testers to test uniformity, identity with a
known distribution, identity with an unknown distribution, and independence of marginals of query complexity
polynomial in the dimension, thus avoiding the curse
of dimensionality.

\subsubsection*{Acknowledgements}
\label{sec:acknowledgements}
 The authors would like to thank the anonymous reviewers for their insightful suggestions and comments, which significantly improved the paper. In particular, the authors would like to thank the first reviewer of the ToCT submission for suggesting the use of Levin's economic work strategy, which resulted in a speedup of all our algorithms by a factor of $n/\epsilon$.

 Rishiraj is supported by \textit{SERB ECR/2017/001974}.

\bibliography{Main}   
\appendix    
 
\section{A Weaker Lower Bound with Simple Proof}
\label{appendix}

\begin{theorem}\label{thm:ind-lb}
For any $0\leq \epsilon \leq 1/2$ any $(\epsilon, 1/3)-\textsc{SubCond}$ Uniformity-Tester has subcube-conditional sample complexity $\Omega(\sqrt[4]{n}/\sqrt{\epsilon})$.
The lower bound holds even for the case when the domain is $\{0,1\}^n$ and the given distribution is a product of $n$ independent (though not necessarily identical) distributions. 
\end{theorem}

\begin{proof}
Let $\mu$ be a product distributions over the domain $\{0,1\}^n$ with marginals $\mu_1, \dots, \mu_n$. 
So $\mu = \mu_1 \otimes \dots \otimes \mu_n$. Note that since the $\mu_i$ are independent, if $i \neq j$ then conditioning on $\mu_i$ 
does not affect the samples we get from a $\mu_j$. Also, since the $\mu_i$ are all distributions over a two-element set (namely $\{0,1\}$), conditioning 
on any subset of $\{0,1\}$ also of no use. Thus drawing subcube-conditional-samples from $\mu$ is as good as drawing samples (without any conditioning) from 
$\mu$. 

So it is sufficient for us to prove that for any $0\leq \epsilon \leq 1/2$ any $(\epsilon, 1/3)$ Uniformity-Tester has sample complexity $\Omega(\sqrt[4]{n})$, when the domain is $\{0,1\}^n$ and the given distributions are product distributions. 

The main idea of the proof is to use a standard technique from property testing where the following lemma is used. The following lemma has been rewritten in 
the language and context of this paper. A proof of the general statement of the lemma can be found in \cite{Fischer04, FNS04}. 

\begin{theorem}\label{yao}
Let $P$ be a property of distributions over $\sigma^n$ that we want to test. 
Suppose $\mathcal{D}_Y$ is a distribution over all the distributions that satisfy the given property $P$, and let $\mathcal{D}_N$ be a distribution 
over all distributions that are $\epsilon$-far from satisfying the property $P$. Let $Q_Y$ be the distribution over outcomes of $q$ samples when the samples are drawn from a distribution $D_Y$ that is drawn according to $\mathcal{D}_Y$. Similarly, let $Q_N$ be the distribution over outcomes of $q$ samples when the samples are drawn from a distribution $D_N$, that is drawn according to the $\mathcal{D}_N$. If the variation distance between $Q_Y$ and $Q_N$ is less than $1/3$, then any $(\epsilon, 1/3)$-Tester for the property $P$ will have sample complexity more than $q$. 
\end{theorem}

In the context of our theorem,  the property $P$ is ``Uniformity". So the distribution $\mathcal{D}_Y$ is the uniform distribution over the domain $\{0,1\}$. Now let us define the distribution $\mathcal{D}_N$:

Let $D_1$ be the distribution over $\{0,1\}$ where $1$ is produced with probability $(1/2+2\sqrt{\frac{\epsilon}{n}})$ and  $0$ produced with probability $(1/2-2\sqrt{\frac{\epsilon}{n}})$. And let $D_0$ be the distribution over $\{0,1\}$ where $1$ is produced with probability $(1/2-2\sqrt{\frac{\epsilon}{n}})$ and  $0$ produced with probability $(1/2+2\sqrt{\frac{\epsilon}{n}})$.

Consider the set of distributions $\mathcal{D}$ over $\{0,1\}^n$ which are a product of $n$ distribution each of which is either $D_0$ or $D_1$. That is, 
$$\mathcal{D} = \left\{ \mu_1\otimes \dots \otimes \mu_n \ \mid \ \mbox{for all i, $\mu_i$ is either $D_0$ or $D_1$} \right\}$$

\begin{claim}\label{cl:far}

  Any $\mu \in \mathcal{D}$ is $\epsilon$-far from uniform. That is, for any $\mu \in \mathcal{D}$ we have 
$$d(\mu, \mathcal{U}) \geq \epsilon$$
\end{claim} 

From Claim~\ref{cl:far} we see that all the distributions in $\mathcal{D}$ are $\epsilon$-far from uniform. Thus we can take the distribution $\mathcal{D}$ as our distribution $\mathcal{D}_N$. If a distribution is drawn from $\mathcal{D}_N$ or $\mathcal{D}_Y$, $q$ samples from the distribution will give $q$ many 
$\{0,1\}$-strings of length $n$. Note that if a distribution is drawn from  $\mathcal{D}_Y$ (that is, the distribution is the uniform distribution over $\{0,1\}^n$), then the distribution of the outcomes of $q$ samples is a uniform distribution over $\{0,1\}^{nq}$. So, by theorem \ref{yao}, it is enough to show that if $\mu$ is drawn from 
$\mathcal{D}_N$ then the distribution of the outcomes (as a distribution over $\{0,1\}^{nq}$) is $1/3$-close to uniform.

Note that $\mu$ is a distribution drawn from $\mathcal{D}_N$ we can think of $\mu$ as $\mu_1\otimes \dots \otimes \mu_n$ where each $\mu_i$ is independently and uniformly chosen from the set $\{D_0, D_1\}$. Let $\mu^q$ be the distribution over $\{0,1\}^{nq}$  when $q$ samples are drawn from 
$\mu$. And now the following lemma completes the proof of Theorem~\ref{thm:ind-lb}. 

\begin{lemma}\label{le:lb} If $q \leq \frac{\sqrt[4]{n}}{20\sqrt{\epsilon}}$ then 
 $$d(\mu^q, \mathcal{U}) \leq \frac{1}{3}.$$
\end{lemma}

% The proof of Claim \ref{cl:far} and Lemma~\ref{le:lb} is given in Appendix \ref{sec:left-proofs-lower}.
\end{proof}

\subsection{Proof of Claim~\ref{cl:far}}
Let $\mu = \mu_1\otimes \dots \otimes \mu_n$. Without loss of generality, we will assume that all the $\mu_i$'s are the distribution $D_1$. That is $1$ is produced with probability $(1/2+2\sqrt{\frac{\epsilon}{n}})$ and  $0$ produced with probability $(1/2-2\sqrt{\frac{\epsilon}{n}})$. For simplifying notations, we will assume $1$ is produced with probability $(1/2+\epsilon')$ and  $0$ produced with probability $(1/2-\epsilon')$.
%We will have to use Hellinger's distance for the proof. 

Since we know $d(\mu, \mathcal{U}) \geq H(\mu, \mathcal{U})^2$, it is enough for us to prove $H(\mu, \mathcal{U})^2 \geq \epsilon$. For any $x\in \{0,1\}^n$ let $p(x)$ be the probability of getting $x$ when drawn from $\mu$. Note that the probability of getting $x$ when drawn from $\mathcal{U}$ is $1/2^n$. 

By definition we have $$H(\mu, \mathcal{U})^2 = \frac{1}{2}\sum_{x\in \{0,1\}^n}\left(\sqrt{p(x)} - \sqrt{1/2^n}\right)^2 = 1 -\sum_{x\in \{0,1\}^n}\left(\sqrt{p(x)/2^n}\right)$$
Now note that if $x$ has $k$ 1's and $(n-k)$ 0's then $p(x) = (1/2 + \epsilon')^k(1/2-\epsilon')^{n-k}$. So we have
$$\sum_{x\in \{0,1\}^n}\left(\sqrt{p(x)/2^n}\right) = \frac{1}{2^n}\sum_{k=0}^n \binom{n}{k}\sqrt{(1 + 2\epsilon')^k(1-2\epsilon')^{n-k}} = \frac{1}{2^n}\left(\sqrt{(1 + 2\epsilon')} + \sqrt{(1 - 2\epsilon')}\right)^n$$

Now since $(\sqrt{1 + x} + \sqrt{1-x}) \leq 2(1 - \frac{x^2}{8})$ for all $x \leq 1$ so,
$$\frac{1}{2^n}\left(\sqrt{(1 + 2\epsilon')} + \sqrt{(1 - 2\epsilon')}\right)^n \leq \left(1 - \frac{\epsilon'^2}{2}\right)^n \leq \left( 1 - \frac{\epsilon'^2n}{2} + \frac{\epsilon'^4}{4}\binom{n}{2}\right).$$

The last inequality follows from the fact that $(1-x)^n \leq (1 - xn + \binom{n}{2}x^2)$.
Now putting all the things together, we have 
$$H(\mu, \mathcal{U})^2 \geq \left(1 - \left(1 - \frac{\epsilon'^2n}{2} + \frac{\epsilon'^4}{4}\binom{n}{2}\right)\right) \geq  \left(\frac{\epsilon'^2n}{2} - \frac{\epsilon'^4}{4}\binom{n}{2}\right)$$

If $\epsilon' = 2\sqrt{\epsilon/n}$ then from the above inequality, and the fact that $\epsilon <1/2$, we have $H(\mu, U)^2 \geq \epsilon$.

\subsection{Proof of Lemma~\ref{le:lb}}

Let us start with a claim. We defer the proof of the claim to the end of this section. 

\begin{claim}\label{cl:individualH}
If $P$ and $Q$ be two distributions over $\Sigma$ and for all $x\in \Sigma$ we have $$\Pr_{P}(x) = (1 + \epsilon_x)\Pr_{Q}(x) $$
then we have 
$$H(P, Q)^2 \leq \frac{1}{2}\sum_{x\in \Sigma}\epsilon_x^2\Pr_{Q}(x)$$
\end{claim}

Claim~\ref{cl:individualH} helps to upper bound the Hellinger distance in terms of the $\ell_{\infty}$ distance. 
Now let $\Sigma =\{0,1\}^q$. And let $\mu_i^q$ be the distribution on $\Sigma$ that is obtained by drawing 
$q$ samples from $\mu_i$. Clearly, $\mu^q = \mu_1^q \otimes \mu_2^q \otimes \dots \otimes \mu_n^q$. 
To prove that the variation distance of $\mu^q$ from uniform is less than $1/3$, we will first show that the $\ell_{\infty}$ distance of  
$\mu_i$ from uniform is small, then using Claim~\ref{cl:individualH} we get that the Hellinger distance of $\mu_i^q$ from uniform is small. 
And then, we can show that if all the $\mu_i^q$ has a small Hellinger distance from uniform, then $\mu^q$ has a small Hellinger distance from uniform, which would give an upper bound on the variation distance of $\mu^q$ from uniform. 

Now the following claim upper bounds the $\ell_{\infty}$ distance of $\mu_i^q$ from uniform.

\begin{claim}\label{cl:individualV}
For all $i$ and for all $x\in \Sigma$  %if $\Pr_{x\leftarrow U[x] >  \Pr_{x\leftarrow \mu_i^q[x]$ then 
$$|\Pr_{\mathcal{U}}(x) - \Pr_{\mu_i^q}(x)| \leq \frac{10\epsilon q^2}{2^qn}$$
Or, in other words, for all $x\in \Sigma$ if
$$\Pr_{\mu_i^q}(x) = (1 \pm \epsilon_x)\Pr_{\mathcal{U}}(x)$$
then $|\epsilon_x| \leq 10\epsilon q^2/n$
\end{claim}

By definition of Hellinger distance and variation distance, we have
$$d(\mu^q, \mathcal{U}) = \sum_{x\in \{0,1\}^{qn}}\left|\Pr_{\mu^q}(x) - \Pr_{\mathcal{U}}(x) \right| \leq 2H(\mu^q, \mathcal{U})$$

Again we know that for any two product distributions $P = P_1 \otimes \dots \otimes P_n$ and $Q= Q_1\otimes \dots \otimes Q_n$ 
$$H(P_1\otimes \dots \otimes P_n, Q_1\otimes \dots \otimes Q_n)^2 \leq \sum_{i=1}^n H(P_i, Q_i)^2. $$

Thus we have 
\begin{equation}\label{eq:productH}
d(\mu^q, \mathcal{U}) \leq 2\sqrt{\left( \sum_{i=1}^n H(\mu_i^q, \mathcal{U})^2 \right)}
\end{equation}

From Equation~\ref{eq:productH} and Claim~\ref{cl:individualH} we have 
$$d(\mu^q, \mathcal{U})  \leq 2 \sqrt{\sum_{i=1}^n \frac{1}{2}\sum_{x\in \Sigma}q(x)\epsilon_x^2},$$
where, $q(x) = \Pr_{\mathcal{U}}(x)$. So $q(x) = 2^q$. From Claim~\ref{cl:individualV} we have  that $\epsilon_x = 10\epsilon q^2/n$.
So we have  

$$d(\mu^q, \mathcal{U})  \leq 2\sqrt{\sum_{i=1}^n  \left(10\epsilon q^2/n\right)^2}$$

Thus if $q \leq \sqrt[4]{n}/20\sqrt{\epsilon}$ we have 
$d(\mu^q, \mathcal{U})  \leq 2\sqrt{1/40}$ which is less than $1/3$

\subsubsection{Proof of Claim~\ref{cl:individualH}}
Let $p(x) = \Pr_{P}(x)$ and $q(x) = \Pr_{Q}(x)$. By definition 
$$H(P, Q)^2 = \frac{1}{2}\sum_{x\in \Sigma}\left(\sqrt{p(x)} - \sqrt{q(x)}\right)^2 = \left(1 - \sum_{x\in \Sigma}\sqrt{p(x)q(x)}\right) $$
Now $\sqrt{p(x)q(x)} = q(x)\sqrt{1 + \epsilon_x}$. Now it is easy to verify that for all $x$ such that $|x|\leq 1$, we have 
$$\sqrt{1+x} \geq 1 + \frac{x}{2} - \frac{x^2}{2}$$
So, from the above observation, 
$$\sqrt{p(x)q(x)}  = q(x)\sqrt{1 + \epsilon_x} \geq q(x)\left(1 + \frac{\epsilon_x}{2} - \frac{\epsilon_x^2}{2}\right)$$
Now since $\sum_x q(x) = 1$ and $\sum_x q_x \epsilon_x = 0$ so we have 
$$H(P, Q)^2 \leq  \left( 1 - \sum_x q(x)\left(1 + \frac{\epsilon_x}{2} - \frac{\epsilon_x^2}{2}\right)\right) = \frac{1}{2}\sum_{x\in \Sigma}q(x)\epsilon_x^2$$

\subsubsection{Proof of Claim~\ref{cl:individualV}}
Let $x\in \Sigma$  has $k$ $1$'s and $(q-k)$ $0$'s. Since the $\mu_i$ is either the distribution $D_1$ with probability $1/2$ or distribution $D_2$ with probability $1/2$, so the probability of $x$ appearing, when drawn from $\mu_i^q$, is

\begin{align*}
  &\frac{1}{2} \left((\frac{1}{2} + \epsilon')^k(\frac{1}{2} - \epsilon')^{q-k} + (\frac{1}{2} - \epsilon')^k(\frac{1}{2} + \epsilon')^{q-k}\right)&\\
 & =  \frac{1}{2^q}\frac{(1 + 2\epsilon')^k(1-2\epsilon')^{q-k} + (1 - 2\epsilon')^k(1+2\epsilon')^{q-k}}{2}&
\end{align*}

\noindent Using the inequality $(1 + x)^r \geq 1 + xr$ (holds for $x\geq -1$ and $r \in \mathbb{N}$), we have
\begin{align*}
  (1 + 2\epsilon')^k(1-2\epsilon')^{n-k} + (1 - 2\epsilon')^k(1+2\epsilon')^{n-k}\geq (1 + 2k\epsilon')(1 - 2(q-k)\epsilon') + (1 - 2k\epsilon')(1 + 2(q-k)\epsilon') 
\end{align*}

\noindent The right-hand side of the above inequality is equal to $(2 - 8k(q-k)\epsilon'^2)$. Thus we have 
$$\Pr_{\mu_i^q}(x) \geq \frac{1}{2^q}(1 - 4k(q-k)\epsilon'^2) \geq  \frac{1}{2^q}\left(1 - \frac{4q^2\epsilon}{n}\right)$$

For the upper bound, we shall use the following inequality. Let $r \in \mathbb{N},x \geq -1$ be such that $xr<1$.  It holds that
$$(1 + x)^r \leq 1 + xr + x^2r^2$$

The above inequality can be easily proved using the following facts.
\begin{enumerate}
\item {\bf When $r \in \mathbb{N}, x >0$ and $rx<1$ }

  \begin{enumerate}
  \item  it holds that $(1+x)^r\leq \mathrm{e}^{rx}$.
    \item  as $0\leq rx < 1$ it holds that $\mathrm{e}^{rx} \leq 1+xr +x^2r^2$.
  \end{enumerate}

\item {\bf When $r \in \mathbb{N}, -1\leq x\leq 0 $} it holds that $(1+x)^r \leq 1+xr+x^2r^2$ (can be proved using induction on $r$).

\end{enumerate}

Since $\epsilon' = 2\sqrt{\epsilon/n}$ and $q \leq \sqrt[4]{n}$, $\epsilon'q <1$. Hence, for all $k\leq q$,

$$(1 + 2\epsilon')^k(1-2\epsilon')^{q-k} \leq (1 + 2k\epsilon' + 4k^2\epsilon'^2)(1 - 2(q-k)\epsilon' + 4(q-k)^2\epsilon'^2),$$
$$(1 - 2\epsilon')^k(1+2\epsilon')^{q-k} \leq (1 - 2k\epsilon' + 4k^2\epsilon'^2)(1 + 2(q-k)\epsilon' + 4(q-k)^2\epsilon'^2)$$
 
 and thus $$\frac{(1 + 2\epsilon')^k(1-2\epsilon')^{q-k} + (1 - 2\epsilon')^k(1+2\epsilon')^{q-k}}{2} \leq \left(1 + 2\epsilon'^2q^2 + q^4\epsilon'^4\right)$$
 
Since $\epsilon' = 2\sqrt{\epsilon/n}$ and $q \leq \sqrt[4]{n}$ so we have   
 $$\left(1 + 2\epsilon'^2q^2 + q^4\epsilon'^4\right) \leq \left( 1 + \frac{10\epsilon q^2}{n}\right).$$
 And thus, we have 
 $$\frac{1}{2^q}\left(1 - \frac{4q^2\epsilon}{n}\right) \leq  \Pr_{x\leftarrow \mu_i^q}(x) \leq \frac{1}{2^q}\left( 1 + \frac{10\epsilon q^2}{n}\right)$$

 % \begin{enumerate}
 % \item 
 % \end{enumerate}

\end{document}